\documentclass[12pt,draftcls,onecolumn]{IEEEtran}
\usepackage{latexsym,url}
\usepackage{amsmath,amssymb,array,graphicx,palatino,verbatim}
\usepackage{setspace}
{\makeatletter
 \gdef\xxxmark{%
   \expandafter\ifx\csname @mpargs\endcsname\relax 
     \expandafter\ifx\csname @captype\endcsname\relax 
       \marginpar{xxx}
     \else
       xxx 
     \fi
   \else
     xxx 
   \fi}
 \gdef\xxx{\@ifnextchar[\xxx@lab\xxx@nolab}
 \long\gdef\xxx@lab[#1]#2{{\bf [\xxxmark #2 ---{\sc #1}]}}
 \long\gdef\xxx@nolab#1{{\bf [\xxxmark #1]}}
}

\let\realbfseries=\bfseries
\def\bfseries{\realbfseries\boldmath}

\newif\ifabstract
\abstracttrue
\newif\iffull
\ifabstract \fullfalse \else \fulltrue \fi

\let\epsilon=\varepsilon

\newcommand{\be}{\begin{equation}}
\newcommand{\ee}{\end{equation}}
\newcommand{\ben}{\begin{equation*}}
\newcommand{\een}{\end{equation*}}
\newcommand{\bsp}{\begin{split}}
\newcommand{\ensp}{\end{split}}

\newtheorem{lemma}{Lemma}
\newtheorem{defi}{Definition}

\newtheorem{prop}{Proposition}

\begin{document}

\title{Typicality Graphs:\\Large Deviation Analysis }
\author{
\authorblockN{ Ali Nazari, Dinesh Krithivasan, S. Sandeep Pradhan, Achilleas Anastasopoulos\\}
\authorblockA{ Dept. of Electrical Engineering and computer science\\
University of Michigan, Ann Arbor\\
E-mail: \{anazari,dineshk,pradhanv,anastas\}@umich.edu\\}
%
\and
\authorblockN{ Ramji Venkataramanan\\}
\authorblockA{ Dept. of Electrical Engineering\\
Stanford University\\
E-mail: vramji@stanford.edu }
\and
%

}

\maketitle
\begin{abstract}
Let $\mathcal{X}$ and $\mathcal{Y}$ be finite alphabets and $P_{XY}$
a joint distribution over them, with $P_X$ and $P_Y$ representing
the marginals. For any $\epsilon > 0$, the set of $n$-length
sequences $x^n$ and $y^n$ that are jointly typical \cite{ckbook}
according to $P_{XY}$ can be represented on a bipartite graph. We
present a formal definition of such a graph, known as a
\emph{typicality} graph,  and study some of its properties.
\end{abstract}

\section{introduction}
The concept of typicality and typical sequences is central to
information theory. It has been used to develop computable
performance limits for several communication problems.

Consider a pair of correlated discrete memoryless information
sources $X$\footnote{We use the following notation throughout this
work. Script capitals $\mathcal{U}$, $\mathcal{X}$, $\mathcal{Y}$,
$\mathcal{Z}$,$\ldots$ denote finite, nonempty sets. To show the
cardinality of a set $\mathcal{X}$, we use $|\mathcal{X}|$. We also
use the letters $P$, $Q$,$\ldots$  for probability distributions on
finite sets, and $U$, $X$, $Y$,$\ldots$ for random variables.} and
$Y$ characterized by a generic joint distribution $p_{XY}$ defined
on the product of two finite sets $\mathcal{X} \times \mathcal{Y}$.
An length $n$  $X$-sequence $x^n$ is typical if the empirical
histogram of $x^n$ is close to $p_X$. A pair of length $n$ sequences
$(x^n,y^n) \in \mathcal{X}^n \times \mathcal{Y}^n$ is said to be
jointly typical if the empirical joint histogram of $(x^n,y^n)$ is
close to the joint distribution $p_{XY}$. The set of all jointly
typical sequence pairs is called the typical set of $p_{XY}$.

Given a sequence length $n$, the typical set can be represented in
terms of the following undirected, bipartite graph. The left
vertices of the graph are all the typical $X$-sequences, and the
right vertices are all the typical $Y$-sequences. From well-known
properties of typical sets, there are (approximately) $2^{nH(X)}$
left vertices and  $2^{nH(Y)}$ right vertices. A left vertex is
connected to a right vertex through an edge if the corresponding $X$
and $Y$-sequences are \emph{jointly} typical. From the properties of
joint typicality, we know that the number of edges in this graph is
roughly $2^{nH(X,Y)}$. Further,  every left vertex (a typical
$X$-sequence) has degree roughly equal to $2^{nH(Y|X)}$, i.e., it is
jointly typical with $2^{nH(Y|X)}$ $Y$-sequences. Similarly, each
right vertex  has degree roughly equal to $2^{nH(X|Y)}$.

  In this paper we formally characterize the typicality graph and look at some
subgraph containment problems. In particular, we answer three
questions concerning the typicality graph:
\begin{itemize}
\item When can we find subgraphs such that the left and right vertices of the subgraph
have specified degrees, say $R'_X$ and $R'_Y$, respectively ?
\item What is the maximum size of subgraphs that are complete, i.e., every left vertex is connected to every right
vertex? One of the main contributions of this paper is a sharp
answer to this question.

\item If we create a subgraph by randomly picking a specified number of left and right vertices,  what is the probability
that this subgraph has far fewer edges than expected?
\end{itemize}

These questions arise in a variety of multiuser communication
problems. Transmitting correlated information over a multiple-access
channel (MAC) \cite{PradhanChoi07}, and communicating over a MAC
with feedback \cite{RamjiMAC09} are two problems where the first
question plays an important role. The techniques used to answer the
second question have been used to develop tighter bounds on the
error exponents of discrete memoryless multiple-access
channels~\cite{nazari08},~\cite{nazari09},~\cite{nazari09-Arx}. The
third question arises in the context of transmitting correlated
information over a broadcast channel \cite{ChoiP08}. Moreover, the
evaluation of performance limits of a multiuser communication
problem can be thought of as characterizing certain properties of
typicality graphs of  random variables associated with the problem.

The paper is organized as follows. Some preliminaries are introduced
in section II. In section III, the typicality graphs are formally
defined and some properties about the number vertices, edges, and
degree conditions are  obtained. The main result of the paper which
is  obtained in section IV.

\section{Preliminaries}
In this section, we provide a concise review of some of the results
available in the literature on the typical sequences,
$\delta$-typical sets and their properties~\cite{ckbook}.

\begin{defi}
A  sequence ${x}^n \in \mathcal{X}^n$ is $X$-typical with constant
$\delta$ if
\begin{enumerate}
\item
$|\frac{1}{n}N(a|x^n)-P_{X}(a)|\leq \delta, \quad \forall a \in
\mathcal{X}$
\item
No $a\in \mathcal{X}$ with $P_X(a)=0$ occurs in $x^n$.
\end{enumerate}
The set of such sequences is denoted $T^n_{\delta}(P_X)$ or
$T^n_{\delta}(X)$, when the distribution being used is unambiguous.
\end{defi}

\begin{defi}
Given a conditional distribution $P_{Y|X}$, a sequence $y^n \in
\mathcal{Y}^n$ is conditionally $P_{Y|X}$-typical with $x^n \in
\mathcal{X}^n$ with constant $\delta$ if
\begin{enumerate}
\item
$|\frac{1}{n}N(a,b|x^n,y^n)-\frac{1}{n}N(a|x^n)P_{Y|X}(b|a)|\leq
\delta, \quad \forall a \in \mathcal{X}, b \in \mathcal{Y}.$
\item
$N(a,b|x^n,y^n) =0 $ whenever $P_{Y|X}(b|a)=0$.
\end{enumerate}
The set of such sequences is denoted $T^n_{\delta}(P_{Y|X}|x^n)$ or
$T^n_{\delta}(Y|x^n)$, when the distribution being used is
unambiguous.
\end{defi}
We will repeatedly use the following results, which we state below
as facts:

\textbf{Fact $1$} \cite[Lemma 2.10]{ckbook}: (a) If $x^n \in
T^n_{\delta}(X)$ and $y^n \in T^n_{\delta'}(Y|x^n)$, then $(x^n,y^n)
\in T^n_{\delta+\delta'}(X,Y)$ and $y^n \in T^n_{(\delta +
\delta')|\mathcal{X}|}(Y)$. \footnote{The typical sets are with
respect to distributions $P_X,P_{Y|X}$ and $P_{XY}$, respectively.}

(b) If $x^n \in T^n_{\delta}(X)$ and $(x^n,y^n) \in
T^n_{\epsilon}(X,Y)$, then $y^n \in T^n_{\delta +\epsilon}(Y|x^n)$.
\\

\textbf{Fact $2$} \cite[Lemma 2.13]{ckbook}  \footnote{The constants
of the typical sets for each $n$, when suppressed, are understood to
be some $\delta_n$ with $\delta_n \to 0$ and $\sqrt{n}\cdot \delta_n
\to \infty$ (delta convention).}: There exists a sequence
$\epsilon_n \to 0$ depending only on $|\mathcal{X}|$ and
$|\mathcal{Y}|$ such that for every joint distribution $P_X\cdot
P_{Y|X}$ on $\mathcal{X} \times \mathcal{Y}$, \be
\label{eq:typ_size}
\begin{split}
\left|\frac{1}{n}\log|T^n(X)|- H(X)\right| & \leq \epsilon_n \\
\left|\frac{1}{n}\log|T^n(Y|x^n)|-H(Y|X)\right| & \leq \epsilon_n,
\quad \forall x^n \in T^n(X).
\end{split}
\ee
\\
The next fact deals with the continuity of entropy with respect to
probability distributions.

\textbf{Fact $3$} \cite[Lemma 2.7]{ckbook} If $P$ and $Q$ are two
distributions on $X$ such that
\[ \sum_{x \in \mathcal{X}}|P(x)-Q(x)| \leq \epsilon \leq \frac{1}{2} \]
then
\[ |H(P)-H(Q)| \leq -\epsilon \log \frac{\epsilon}{|\mathcal{X}|} \]

%
%
%

\section{Typicality graphs}
Consider any joint distribution $P_X \cdot P_{Y|X}$ on $\mathcal{X}
\times \mathcal{Y}$.
\begin{defi} \label{def:typ_graph}
For any $\epsilon_{1n},\epsilon_{2n},\lambda_n \to 0$,  the sequence
of  typicality graphs $G_n(\epsilon_{1n},\epsilon_{2n},\lambda_n)$
is defined as follows. For every $n$, $G_n$ is a bipartite graph,
with its left vertices consisting of all $x^n \in
T^n_{\epsilon_{1n}}(X)$ and the right vertices consisting of  all
$y^n \in T^n_{\epsilon_{2n}}(Y)$. A vertex on the left (say
$\tilde{x}^n$) is connected to a vertex on the right (say
$\tilde{y}^n$) iff $(\tilde{x}^n,\tilde{y}^n) \in
T^n_{\lambda_n}(X,Y)$.
\end{defi}

\textbf{Remark}. Henceforth, we will assume that the sequences
$\epsilon_{1n},\epsilon_{2n},\lambda{n}$ satisfy the `delta
convention' \cite[Convention 2.11]{ckbook}, i.e.,
\[ \epsilon_{1n} \to 0, \quad \sqrt{n} \cdot \epsilon_{1n} \to \infty \text{ as } n \to \infty \]
with similar conditions for $\epsilon_{2n}$ and $\lambda_n$ as well.
The delta convention ensures that the typical sets have `large
probability'.

We will use the notation $V_X(.),V_Y(.)$ to denote the vertex sets
of any bipartite graph. Some properties of the typicality graph:
\begin{enumerate}
\item From Fact $2$, we know that for any sequence of typicality graphs $\{G_n(\epsilon_{1n},\epsilon_{2n},\lambda_n)\}$, the cardinality of the vertex sets satisfies
\be \left|\frac{1}{n}\log |V_X(G_n)|- H(X)\right|  \leq \epsilon_n,
\quad \left|\frac{1}{n}\log |V_Y(G_n)|- H(Y)\right|  \leq \epsilon_n
\ee for some sequence $\epsilon_n \to 0$.
\item  The degree of each each vertex $i \in V_X(G_n)$ and $j \in V_Y(G_n)$ satisfies
\be \label{eq:deg_ub} \text{degree}(x^n) \leq
2^{n(H(Y|X)+\epsilon_{n})}, \quad \forall  x^n \in V_X(G_n);\qquad
\text{degree}(y^n) \leq 2^{n(H(X|Y)+\epsilon_{n})}, \quad \forall
y^n \in V_Y(G_n) \ee for some $\epsilon_n \to 0$.
 \proof If $x^n \in
T_{\epsilon_{1n}}^n(X)$ and $(x^n,y^n) \in T_{\lambda_n}^n(X,Y)$,
then from Fact $1$(b),  $y^n \in
T_{\epsilon_{1n}+\lambda_n}^n(Y|x^n)$. From the second part of Fact
$2$, we know that there exists a sequence $\epsilon_n \to 0$ such
that
\be \label{eq:cond_typ_ub}
\left|T^n_{\epsilon_{1n}+\lambda_n}(Y|x^n)\right| \leq
2^{n(H(Y|X)+\epsilon_n)} \ee From this we conclude that
$\text{degree}(x^n) \leq 2^{n(H(Y|X)+\epsilon_{n})}, \forall x^n \in
V_X(G_n)$. An identical argument yields $\text{degree}(y^n) \leq
2^{n(H(X|Y)+\epsilon_{n})},\forall y^n \in V_Y(G_n)$.
\end{enumerate}

Property $2$ gives upper bounds on the degree of each vertex in the
typicality graph. Since we have not imposed any relationships
between the typicality constants $\epsilon_{1n},\epsilon_{2n}$ and
$\lambda_n$, in general it cannot be said that the degree of
\emph{every} $X$-vertex (resp. $Y$-vertex) is close to $2^{NH(Y|X)}$
(resp. $2^{NH(X|Y)}$). However, such an assertion holds for
\emph{almost} every vertex in $G_n$ . Specifically, we can show that
the above degree conditions hold for a subgraph with exponentially
the same size as $G_n$.
\begin{prop} \label{prop:typ_prop}
Every sequence of typicality graphs
$G_n(\epsilon_{1n},\epsilon_{2n},\lambda_n)$ has a sequence of
subgraphs $A_n(\epsilon_{1n},\epsilon_{2n},\lambda_n)$ satisfying
the following properties for some $\delta_n \to 0$.
\begin{enumerate}
\item The vertex set sizes  $|V_X(A_n)| \text{ and } |V_Y(A_n)|$, denoted $\theta_X^n$ and $\theta_Y^n$, respectively, satisfy
\[ \left|\frac{1}{n}\log \theta_X^n- H(X)\right|  \leq \delta_n, \quad \left|\frac{1}{n}\log \theta_Y^n- H(Y)\right|  \leq \delta_n \quad \forall n \]
\item The degree of each $X$-vertex $x^n$, denoted $\theta^{'n}(x^n)$ satisfies
\[ \left|\frac{1}{n}\log \theta^{'n}(x^n)- H(Y|X)\right|  \leq \delta_n  \quad \forall x^n \in V_X(A_n). \]
\item The degree of each $Y$-vertex $y^n$, denoted $\theta^{'n}(y^n)$, satisfies
\[ \left|\frac{1}{n}\log \theta^{'n}(y^n)- H(X|Y)\right|  \leq \delta_n \quad \forall y^n \in V_Y(A_n). \]
\end{enumerate}
\end{prop}

\begin{proof}
The vertex sets $V_X(G_n)$ and $V_Y(G_n)$ are the
$\epsilon_{1n}$-typical and $\epsilon_{2n}$-typical sets of $P_X$
and $P_Y$, respectively. To define the subgraphs $A_n$, we would
like to choose the sequences with type $P_X$ and $P_Y$, respectively
as the vertex sets of the subgraph, with an edge connecting two
sequences if they have joint type $P_{XY}$. However, the values
taken by the joint pmfs $P_{XY}, P_X,P_Y$ may be any real number
between $0$ and $1$, whereas the joint type of two $n$-sequences is
always a rational number(with denominator $n$).  So we choose the
subgraph $A_n$ as follows:
\begin{itemize}
\item For each $n$, approximate the values of $P_{XY}$ to rational numbers with denominator $n$ to obtain pmf $\tilde{P}_{XY}$, respectively.
Clearly $\tilde{P}_{XY}$ is a  valid joint type of length $n$ and
the maximum approximation error is bounded by $\frac{1}{n}$. In
fact, $\forall (x,y)$, we have for all sufficiently large $n$:
    \be |P_{XY} (x,y)- \tilde{P}_{XY} (x,y)| <   \frac{1}{n}  <<  \frac{1}{\sqrt{n}} < \lambda_{n},  \label{eq:xy_approx}\ee
    where the last inequality follows from the delta convention. Using Fact $1$, we also have
    \begin{gather}
    |P_{X} (x)- \tilde{P}_{X} (x)| <  |\mathcal{Y}|\cdot \frac{1}{n}  <<  \frac{1}{\sqrt{n}} < \epsilon_{1n} \label{eq:x_approx}\\
    |P_{Y} (y)- \tilde{P}_{Y} (y)| <  |\mathcal{X}|\cdot \frac{1}{n}  <<  \frac{1}{\sqrt{n}} < \epsilon_{2n} \label{eq:y_approx}
    \end{gather}

\item The left vertex set of $A_n$ is $T^n_0(\tilde{P}_X)$, i.e.,  the set of $x^n$ sequences with type $\tilde{P}_X$. The right vertex set of $A_n$ is $T^n_0(\tilde{P}_Y)$- the set of $y^n$ sequences with type $\tilde{P}_Y$. A vertex in $V_X(A_n)$, say ${a}^n$ is connected to a vertex in
$V_Y(A_n)$, say ${b}^n$ iff $({a}^n,{b}^n) \in
T^n_0(\tilde{P}_{X,Y})$, i.e.,  $({a}^n,{b}^n)$ have joint type
$\tilde{P}_{XY}$.
\end{itemize}

From \eqref{eq:xy_approx},\eqref{eq:x_approx} and
\eqref{eq:y_approx}, we have
\begin{gather*}
T^n_0(\tilde{P}_X) \subset T^n_{\epsilon_{1n}}(P_X), \quad T^n_0(\tilde{P}_Y) \subset T^n_{\epsilon_{2n}}(P_Y) \quad \text{and}\\
 T^n_0(\tilde{P}_{X,Y}) \subset T^n_{\lambda_{n}}(P_{X,Y}).
\end{gather*}
Hence $A_n$ is a subgraph of $G_n$, as required.

From \cite[Lemma 2.3]{ckbook}, we have \be \left|\frac{1}{n}\log
|T^n_0(\tilde{P}_X)|- H(\tilde{P}_X)\right|  \leq \delta_{1n}, \quad
\left|\frac{1}{n}\log |T^n_0(\tilde{P}_Y)|- H(\tilde{P}_Y)\right|
\leq \delta_{2n} \quad \forall n, \ee where
$\delta_{1n}=(n+1)^{-|\mathcal{X}|}$ and
$\delta_{2n}=(n+1)^{-|\mathcal{Y}|}$. Fact $3$ establishes the
continuity of entropy with respect to the probability distribution.
Using Fact $3$ along with \eqref{eq:xy_approx},\eqref{eq:x_approx}
and \eqref{eq:y_approx}, we obtain \be \label{eq:vertex_size}
\left|\frac{1}{n}\log |T^n_0(\tilde{P}_X)|- H({P}_X)\right|  \leq
\delta_{1n}, \quad \left|\frac{1}{n}\log |T^n_0(\tilde{P}_Y)|-
H({P}_Y)\right|  \leq \delta_{2n} \quad \forall n, \ee where we have
reused $\delta_{1n},\delta_{2n}$ with some abuse of notation. This
proves the first property.

We now note that  $x^n \in V_X(A_n)=T^n_0(\tilde{P}_X)$ and $y^n \in
T^n_0(\tilde{P}_{Y|X}|x^n)$ implies a)$(x^n,y^n) \in
T^n_0(\tilde{P}_{X,Y})$ and b)$y^n \in T^n_0(\tilde{P}_Y)=V_Y(A_n)$
(Fact $1$). This implies \be \label{eq:xdeg} \text{degree}(x^n) \geq
|T^n_0(\tilde{P}_{Y|X}|x^n)|, \forall x^n \in V_X(A_n). \ee From
\cite[Lemma 2.5]{ckbook}, we know that \be |T^n_0(\tilde{P}_{Y|X})|
\geq 2^{n(H(\tilde{P}_{Y|X})-\delta_{3n})} \ee where
$\delta_{3n}=|\mathcal{X}||\mathcal{Y}|\frac{\log (n+1)}{n}$. In the
above, $H(\tilde{P}_{Y|X})$ stands for $H(Y|X)$ computed under the
joint distribution $\tilde{P}_{XY}$. Combining this with
\eqref{eq:xdeg}, we get a lower bound on the degree of each $x^n \in
V_X(A_n)$: \be \label{eq:tildebound} \text{degree}(x^n) \geq
2^{n(H(\tilde{P}_{Y|X})-\delta_{3n})} \ee
 From \eqref{eq:xy_approx} and \eqref{eq:x_approx}, one can deduce that $\forall x,y$
 \[ |P_{Y|X} (y|x)- \tilde{P}_{Y|X} (y|x)| < \gamma_n\]
 for some $\gamma_n \to 0$. Combining this with Fact $3$,  \eqref{eq:tildebound} can be written as
\be \text{degree}(x^n) \geq 2^{n(H({P}_{Y|X})-\delta_{3n})}, \ee
where we reuse the symbol $\delta_{3n}$.

Further, \eqref{eq:deg_ub} gives an upper bound on the degree of
each vertex in $G_n$. Hence we have \be \left|\frac{1}{n}\log
\theta^{'n}(x^n)- H(Y|X)\right|  \leq \max(\delta_{3n},\epsilon_n)
\quad \forall x^n \in V_X(A_n) \ee Similarly, we can bound the
degree of each vertex in $V_Y(A_n)$ as \be \left|\frac{1}{n}\log
\theta^{'n}(y^n)- H(X|Y)\right|  \leq \max(\delta_{4n},\epsilon_n)
\quad \forall y^n \in V_Y(A_n) \ee
Finally, we can set
$\delta_n=\max(\delta_{1n},\delta_{2n},\delta_{3n},\delta_{4n},\epsilon_n)$
to complete the proof of the proposition.
\end{proof}

\section{Sub-graphs contained in typicality graphs}
In this section, we study the subgraphs contained in a sequence of
typicality graphs.
\subsection{Subgraphs of general degree} \label{subsec:subgraphsgendegree}
%
%
%
\begin{defi} \label{def:gen_subg}
A sequence of typicality graphs
$G_n(\epsilon_{1n},\epsilon_{2n},\lambda_n)$ is said to contain a
sequence of subgraphs $\Gamma_n$ of rates $(R_X,R_Y, R'_X, R'_Y)$ if
for each $n$, if there exists a sequence $\delta_n \to 0$ such that
\begin{enumerate}
\item
 The vertex sets of the subgraphs have sizes (denoted $\Delta_X^n$ and $\Delta_Y^n$) that satisfy
\[ \left|\frac{1}{n}\log \Delta_X^n- R_X\right|  \leq \delta_n, \quad \left|\frac{1}{n}\log \Delta_Y^n- R_Y\right|  \leq \delta_n, \: \forall n.\]
\item The degree of each vertex  $x^n$ in $V_X(\Gamma_n)$, denoted $\Delta^{'n}(x^n)$ satisfies
\[ \left|\frac{1}{n}\log\Delta^{'n}(x^n)- R'_Y \right| \leq \delta_n, \quad \forall x^n \in V_X(\Gamma_n),\: \forall n . \]
\item The degree of each vertex  $y^n$ in the $V_Y(\Gamma_n)$, denoted $\Delta^{'n}(y^n)$ satisfies
\[ \left|\frac{1}{n}\log\Delta^{'n}(y^n) - R'_X\right| \leq \delta_n, \quad \forall y^n \in V_Y(\Gamma_n),\: \forall n.   \]
\end{enumerate}
\end{defi}
%
%
%
The following proposition gives a characterization of the rate-tuple
of a sequence of subgraphs in the sequence of typicality graphs of
$P_{XY}$.
\begin{prop} \label{prop:gen_rates}
Let $G_n(\epsilon_{1n},\epsilon_{2n},\lambda_n)$ be a sequence of
typicality graphs of $P_{XY}$. Define
\[ \mathcal{R} \triangleq \{ (R_X,R_Y,R'_X,R'_Y):  G_n(\epsilon_{1n},\epsilon_{2n},\lambda_n) \text{ contains subgraphs of rates }
(R_X,R_Y,R'_X,R'_Y)\}  \]Then \be \label{eq:gen_subgr_rates}
\mathcal{R} \supseteq \{(R_X,R_Y,R'_X,R'_Y): R_X \leq H(X|U), \: R_Y
\leq H(Y|U), \: R'_X \leq H(Y|XU) , \: R'_Y \leq H(Y|XU) \text{ for
some } P_{U|XY}.\} \ee
\end{prop}
\begin{proof}  

\textbf{Definition of $\Gamma_n$}. Consider any conditional
distribution $P_{U|XY}$. This fixes the joint distribution
$P_{XYU}=P_{XY}P_{U|XY}$. We construct $\Gamma_n$ as follows.
\begin{itemize}
\item
For each $n$, approximate the values of $P_{UXY}$ to rational
numbers with denominator $n$ to obtain pmf $\tilde{P}_{UXY}$,
respectively. Clearly $\tilde{P}_{UXY}$ is a  valid joint type of
length $n$ and  the maximum approximation error is bounded by
$\frac{1}{n}$.
    Marginalizing the joint pmf, we also have $\forall x,y$
    \begin{gather}
     |P_{XY} (x,y)- \tilde{P}_{XY} (x,y)| < |\mathcal{U}|\cdot  \frac{1}{n} <<  \frac{1}{\sqrt{n}} < \lambda_{n}, \label{eq:u_xy_approx}\\
    |P_{X} (x)- \tilde{P}_{X} (x)| <  |\mathcal{Y}|\cdot|\mathcal{U}|\cdot \frac{1}{n}  <<  \frac{1}{\sqrt{n}} < \epsilon_{1n} \label{eq:u_x_approx}\\
    |P_{Y} (y)- \tilde{P}_{Y} (y)| <  |\mathcal{X}|\cdot|\mathcal{U}|\cdot \frac{1}{n}  <<  \frac{1}{\sqrt{n}} < \epsilon_{2n}, \label{eq:u_y_approx}
    \end{gather}
    where the last inequality in each equation follows from the delta
    convention. Further $\forall u$
    \be |P_{U} (u)- \tilde{P}_{U} (u)| <  |\mathcal{Y}|\cdot|\mathcal{X}|\cdot \frac{1}{n}. \label{eq:uapprox}\ee

\item
Pick any length $n$ sequence $u^n$  with type $\tilde{P}_{U}$, i.e.,
$u^n \in T_{0}^n(\tilde{P}_U)$. Consider a bipartite graph
$\Gamma_n$ with $X$-vertices consisting of all $x^n \in
T^n_{0}(\tilde{P}_{X|U}|u^n)$, $Y$-vertices consisting of all $y^n
\in T_{0}^n(\tilde{P}_{Y|U}|u^n)$. In other words, having fixed
$u^n$, the $X$-vertex sets and $Y$-vertex sets consist of all length
$n$ sequences having conditional type $\tilde{P}_{X|U}$ and
$\tilde{P}_{Y|U}$, respectively. Vertices $x^n \in V_X(\Gamma_n)$
and $y^n \in V_Y(\Gamma_n)$ are connected in $\Gamma_n$ iff
$(x^n,y^n) \in T^n_0(\tilde{P}_{XY|U}|u^n)$, i.e., if they have the
conditional joint type $P_{XY|U}$ given $u^n$.
\end{itemize}

Let us verify that $\Gamma_n$ is a subgraph of $G_n$. From Fact $1$,
if $u^n \in T_{0}^n(\tilde{P}_U)$ and $x^n \in
T_{0}^n(\tilde{P}_{X|U}|u^n)$, then $(x^n,u^n) \in
T^n_{0}(\tilde{P}_{X,U})$. Consequently, $x^n \in
T^n_{0}(\tilde{P}_{X})$. Similarly,  all $y^n \in
T_{0}^n(\tilde{P}_{Y|U}|u^n)$  belong to $T^n_{0}(\tilde{P}_Y)$.
 On the same lines, if $u^n \in T_{0}^n(\tilde{P}_U)$ and
$(x^n,y^n) \in T_{0}^n(\tilde{P}_{XY|U}|u^n)$, then $(x^n,y^n,u^n)
\in T^n_{0}(\tilde{P}_{X,Y,U})$.  This implies $(x^n,y^n) \in
T^n_{0}(\tilde{P}_{X,Y})$.
Further, from \eqref{eq:u_xy_approx},\eqref{eq:u_x_approx} and
\eqref{eq:u_y_approx}, we know
\begin{gather*}
T^n_0(\tilde{P}_X) \subset T^n_{\epsilon_{1n}}(P_X)=V_{X}(G_n), \quad T^n_0(\tilde{P}_Y) \subset T^n_{\epsilon_{2n}}(P_Y)=V_Y(G_n) \quad \text{and}\\
 T^n_0(\tilde{P}_{X,Y}) \subset T^n_{\lambda_{n}}(P_{X,Y}).
\end{gather*}
Hence for all sufficiently large $n$, $\Gamma_n$ is a subgraph of
the typicality graph $G_n$.

\textbf{Properties of $\Gamma_n$}. From \cite[Lemma 2.3]{ckbook}, we
have \be \label{eq:vertex_size_gamma} \left|\frac{1}{n}\log
|T^n_0(\tilde{P}_{X|U}|u^n)|- H(\tilde{P}_{X|U})\right|  \leq
\delta_{1n}, \quad \left|\frac{1}{n}\log
|T^n_0(\tilde{P}_{Y|U}|u^n)|- H(\tilde{P}_{Y|U})\right|  \leq
\delta_{2n} \quad \forall n, \ee where
$\delta_{1n}=(n+1)^{-|\mathcal{X}||\mathcal{U}|}$ and
$\delta_{2n}=(n+1)^{-|\mathcal{Y}||\mathcal{U}|}$. Using
\eqref{eq:u_x_approx}, \eqref{eq:u_y_approx} with
\eqref{eq:uapprox}, we know that $\tilde{P}_{X|U}, \tilde{P}_{Y|U}$
are close to ${P}_{X|U},{P}_{Y|U}$, respectively. Using Fact $3$, we
know that the entropies $H(\tilde{P}_{X|U}), H(\tilde{P}_{Y|U})$
must close to $H(P_{X|U}),H(P_{Y|U})$, respectively. Thus we can
write \eqref{eq:vertex_size_gamma} as (reusing
$\delta_{1n},\delta_{2n}$) \be \left|\frac{1}{n}\log
|T^n_0(\tilde{P}_{X|U}|u^n)|- H({P}_{X|U})\right|  \leq \delta_{1n},
\quad \left|\frac{1}{n}\log |T^n_0(\tilde{P}_{Y|U}|u^n)|-
H({P}_{Y|U})\right|  \leq \delta_{2n} \quad \forall n, \ee Thus, the
vertex sets of $\Gamma_n$ have rates $R_X=H(X|U)$ and $R_Y=H(Y|U)$,
as required.

Using Fact $1$, for any $x^n \in V_X(\Gamma_n)$, every $y^n \in
T^n_0(\tilde{P}_{Y|XU}|x^n,u^n)$ will satisfy  a) $(x^n,y^n) \in
T^n_0(\tilde{P}_{X,Y|U}|u^n)$ and b) $y^n \in
T^n_0(\tilde{P}_{Y|U}|u^n)$. Hence \be \label{eq:degyxu_lb}
\text{degree}(x^n) \geq |T^n_0(\tilde{P}_{Y|XU}|x^n,u^n)| \geq
2^{n(H(\tilde{P}_{Y|XU})-\delta_{3n})}, \ee where $\delta_{3n} =
|\mathcal{X}||\mathcal{Y}||\mathcal{U}|\frac{\log(n+1)}{n}$. We can
also upper bound the degree of $x^n$ by noting that $x^n \in
T^n_0(\tilde{P}_{X|U}|u^n)$ and $(x^n,y^n) \in
T^n_0(\tilde{P}_{X,Y|U}|u^n)$ implies $y^n \in
T^n_0(\tilde{P}_{Y|XU}|x^n,u^n)$. From \cite[Lemma 2.5]{ckbook},
\[ |T^n_0(\tilde{P}_{Y|XU}|x^n,u^n)| \leq 2^{nH(\tilde{P}_{Y|XU})}.\]
Combining this with \eqref{eq:degyxu_lb}, we have \be
\left|\frac{1}{n}\log\Delta^{'n}(x^n)- H(\tilde{P}_{Y|XU}) \right|
\leq \delta_{3n}, \quad \forall x^n \in V_X(\Gamma_n),\: \forall n .
\ee In a similar fashion, we can show that \be
\left|\frac{1}{n}\log\Delta^{'n}(y^n)- H(\tilde{P}_{X|YU}) \right|
\leq \delta_{4n}, \quad \forall y^n \in V_Y(\Gamma_n),\: \forall n .
\ee
Since the distributions $\tilde{P}_{Y|XU}$ and $\tilde{P}_{X|YU}$
are close to $P_{Y|XU}$ and $P_{X|YU}$, respectively, Fact $3$
enables us to replace $H(\tilde{P}_{Y|XU}), H(\tilde{P}_{X|YU})$
with $H({P}_{Y|XU}), H({P}_{X|YU})$, respectively in the two
preceding equations.

Taking
$\delta_n=\max(\delta_{1n},\delta_{2n},\delta_{3n},\delta_{4n})$, we
have shown the existence of a sequence of subgraphs $\Gamma_n$ with
rates $(H(X|U),H(Y|U),H(Y|XU),H(X|YU))$. Since we can simply exclude
edges from $\Gamma_n$ to obtain subgraphs with smaller rates, it is
clear that all rate tuples characterized by \[ (R_X,R_Y,R'_X,R'_Y):
R_X \leq H(X|U), \: R_Y \leq H(Y|U), \: R'_X \leq H(Y|XU) , \: R'_Y
\leq H(Y|XU)  \] are achievable for every conditional distribution
$P_{U|XY}$.
\end{proof}

\subsection{Nearly complete subgraphs} \label{subsec:subgraphsnearlycomplete}
A complete bipartite graph is one in which each vertex of the first
set is  connected with every vertex on the other set. We next
consider a specific class of subgraphs, namely nearly complete
subgraphs. For this class of subgraphs, we have a converse result
that fully characterizes the set of nearly complete subgraphs
present in any typicality graph.
\begin{defi} \label{def:nc_subg}
A sequence of typicality graphs $G_n(\epsilon_{1n},\epsilon_{2n},\lambda_n)$ is said to contain a sequence of  nearly complete subgraphs $\Gamma_n(\epsilon_{1n},\epsilon_{2n},\lambda_n)$ of rates $(R_X,R_Y)$ if for each $n$, if there exists a sequence $\delta_n \to 0$ 
such that
\begin{enumerate}
\item
The sizes of the vertex sets of the subgraphs, denoted $\Delta_X^n$
and $\Delta_Y^n$, satisfy
\[ \left|\frac{1}{n}\log \Delta_X^n- R_X\right|  \leq \delta_n, \quad \left|\frac{1}{n}\log \Delta_Y^n- R_Y\right|  \leq \delta_n, \: \forall n.\]
\item The degree of each vertex  $x^n$ in the $X$-set, denoted $\Delta^{'n}(x^n)$ satisfies
\[ \frac{1}{n}\log\Delta^{'n}(x^n) \geq R_Y-\delta_n, \quad \forall x^n \in V_X(\Gamma_n),\: \forall n . \]
\item The degree of each vertex  $j$ in the $Y$-set, denoted $\Delta^{'n}_{j}$ satisfies for all n
\[ \frac{1}{n}\log\Delta^{'n}(y^n) \geq R_X-\delta_n, \quad \forall y^n \in V_Y(\Gamma_n),\: \forall n.   \]
\end{enumerate}
\end{defi}
%
%
%
\begin{prop} \label{prop:full_con}
Let $G_n(\epsilon_{1n},\epsilon_{2n},\lambda_n)$ be a sequence of
typicality graphs for $P_{XY}$. Define
\[ \mathcal{R} \triangleq \{ (R_X,R_Y):  G_n(\epsilon_{1n},\epsilon_{2n},\lambda_n) \text{ contains nearly complete subgraphs of rates }(R_X,R_Y)\}  \]Then
\begin{enumerate}
\item
\be \label{eq:subgr_rates} \mathcal{R}\supseteq \{ (R_X,R_Y): R_X
\leq H(X|U), \: R_Y \leq H(Y|U) \text{ for some } P_{U|XY} \text{
s.t. } X-U-Y \}\footnote{$X,U,Y$ form a Markov chain, in that
order.} \ee
\item
For all sequences of nearly complete subgraphs of $G_n$  such that
the sequence $\delta_n$ (in Definition \ref{def:nc_subg})  converges
to $0$  faster than $1/\log n$ (more precisely,
$\delta_n=o(\frac{1}{\log n})$ or $\lim_{n\to \infty}\delta_n \log
n=0$), the rates of the subgraph $(R_X,R_Y)$ satisfy
\[ R_X \leq H(X|U), \: R_Y \leq H(Y|U) \text{ for some } P_{U|XY} \text{ s.t. } X-U-Y \]
\end{enumerate}
\end{prop}
\begin{proof} The first part of the proposition follows directly from Proposition \ref{prop:gen_rates} by choosing $P_{U|XY}$ such that $X-U-Y$ form a Markov chain.
We now prove the converse under the stated assumption that  the
sequence $\delta_n$ satisfies $\lim_{n\to \infty}\delta_n \log n=0$.

Suppose that a sequence of typicality graphs
$G_n(\epsilon_{1n},\epsilon_{2n},\lambda_n)$ contains nearly
complete subgraphs $\Gamma_n$ of rates $R_X,R_Y$. The total number
of edges in $\Gamma_n$ can be lower bounded as \be
\begin{split}
|\text{Edges}(\Gamma_n)| & \geq \Delta_X^n \cdot \text{ minimum degree of a vertex in }V_{X}({\Gamma}_n)\\
& \geq \Delta_X^n \cdot 2^{n(R_Y-\delta_n)}\\
& \geq \Delta_X^n \cdot 2^{n(R_Y-\delta_n)} \Delta_Y^n \cdot 2^{-n(R_Y+\delta_n)} \\
&= \Delta_X^n \cdot \Delta_Y^n \cdot 2^{-2n\delta_n}.
\end{split}
\ee Each of these edges represent a pair $(x^n,y^n)$ that is jointly
$\lambda_n$-typical with respect to the distribution $P_{XY}$. In
other words, each of these pairs $(x^n,y^n)$ belongs to a joint
type\cite{ckbook} that is `close' to $P_{XY}$. Since the number of
joint types of a pair of sequences of length $n$ is at most
$(n+1)^{|\mathcal{X}||\mathcal{Y}|}$, the number of edges belonging
to the dominant joint type, say $\bar{P}_{XY}$ satisfies
 \be \label{eq:an_size1}
 |\text{Edges}(\Gamma_n) \text{ having joint type
}{\bar{P}_{XY}}| \geq \frac{\Delta_X^n \cdot \Delta_Y^n
2^{-2n\delta_n}}{(n+1)^{|\mathcal{X}||\mathcal{Y}|}}. \ee
Define a subgraph $\mathcal{A}_n$ of $\Gamma_n$ consisting only of
the edges having joint type $\bar{P}_{XY}$. A word about the
notation used in the sequel: We will use $i,j$  to index the
vertices in $V_X(\Gamma_n), V_Y(\Gamma_n)$, respectively. Thus $i
\in \{1,\ldots,\Delta_X^n\}$ and $j \in \{1,\ldots,\Delta_Y^n\}$.
The actual sequences corresponding to these vertices will be denoted
$x^n(i),y^n(j)$ etc. Using this notation,
\be %
\mathcal{A}_n \triangleq \{(i,j): i \in V_X(\Gamma_n), j \in
V_Y(\Gamma_n)
\text{ s.t.  $(x^n(i),y^n(j))$ has joint type } \bar{P}_{XY}%
\ee
From \eqref{eq:an_size1},
\be \label{eq:an_size}%
|\mathcal{A}_n| \geq \frac{\Delta_X^n \cdot \Delta_Y^n
2^{-2n\delta_n}}{(n+1)^{|\mathcal{X}||\mathcal{Y}|}} %
\ee
We will prove the converse result using a series of lemmas
concerning $\mathcal{A}_n$. Some of the lemmas are similar to those
required to prove in \cite[Theorem 1]{nazari08}. We only sketch the
proofs of such lemmas, referring the reader to \cite{nazari08} for
details.

 Define random variables $X'^n,Y'^n$ with pmf
\be%
\text{Pr}((X'^n, Y'^n)=(x^n(i),y^n(j))=\frac{1}{|\mathcal{A}_n|}, \:
\text{ if } (i,j)\in \mathcal{A}_n.%
\ee
\begin{lemma} \label{lem:ixybound}
$I(X'^n;Y'^n)\leq 2n\delta_n + |\mathcal{X}||\mathcal{Y}|\log(n+1).$
\end{lemma}
\proof Follow steps  similar to the proof of \cite[Lemma
1]{nazari08}, using \eqref{eq:an_size} to lower bound the size of
$\mathcal{A}_n$.
%

The next lemma is Ahlswede's version of the `wringing' technique.
Roughly speaking, if  it is known that the mutual information
between two random sequences is small, then the lemma gives an upper
bound on the per-letter mutual information terms (conditioned on
some values).
\begin{lemma}\cite{Ahls82} \label{lem:wringing}
Let $A^n$, $B^n$ be RV's with values in $\mathcal{A}^n$,
$\mathcal{B}^n$ resp. and assume that
\begin{equation*}
I(A^n ; B^n) \leq \sigma
\end{equation*}
Then, for any $0 < \delta < \sigma$ there exist $t_1,t_2,...,t_k \in
\{1,...,n\}$ where $0 \leq k < \frac{2 \sigma}{\delta}$ such that
for some
$\bar{a}_{t_1},\bar{b}_{t_1},\bar{a}_{t_2},\bar{b}_{t_2},...,\bar{a}_{t_k},\bar{b}_{t_k}$
\begin{align}
I(A_t;B_t|A_{t_1}=\bar{a}_{t_1},B_{t_1}=\bar{b}_{t_1},...,A_{t_k}=\bar{a}_{t_k},B_{t_k}=\bar{b}_{t_k})
\leq \delta \qquad \text{for } t=1,2,...,n
\end{align}
and
\begin{flalign}
Pr(A_{t_1}=\bar{a}_{t_1},B_{t_1}=\bar{b}_{t_1},...,A_{t_k}=\bar{a}_{t_k},B_{t_k}=\bar{b}_{t_k})\geq
(\frac{\delta}{|\mathcal{A}||\mathcal{B}|(2 \sigma-\delta)})^k.
\end{flalign}
\end{lemma}
\end{proof}

In our case, we will apply Lemma \ref{lem:wringing} to random
variables $X'^n$ and $Y'^n$. Lemma \ref{lem:ixybound} indicates
$\sigma= 2n\delta_n + |\mathcal{X}||\mathcal{Y}|\log(n+1)$, and
$\delta$ shall be specified later.
Hence we have that for some
\[
k \leq \frac{2\sigma}{\delta} = \frac{2(n\delta_n +
|\mathcal{X}||\mathcal{Y}|\log(n+1))}{\delta},
\]
there exist
$\bar{x}_{t_1},\bar{y}_{t_1},\bar{x}_{t_2},\bar{y}_{t_2},...,\bar{x}_{t_k},\bar{y}_{t_k}$
such that \be I(X'_t;
Y'_t|X'_{t_1}=\bar{x}_{t_1},Y'_{t_1}=\bar{y}_{t_1},...,X'_{t_k}=\bar{x}_{t_k},Y'_{t_k}=\bar{y}_{t_k})
\leq \delta \qquad \text{for } t=1,2,...,n. \ee We now define a
subgraph of $\mathcal{A}_n$ consisting of all edges $(X'^n,Y'^n)$
that have
\[X'_{t_1}=\bar{x}_{t_1},Y'_{t_1}=\bar{y}_{t_1},...,X'_{t_k}=\bar{x}_{t_k},Y'_{t_k}=\bar{y}_{t_k} \]
The subgraph denoted as $\mathcal{\bar{A}}_n$  is given by:
\footnote{The heirarchy of subgraphs is $G_n \supset \Gamma_n
\supset \mathcal{A}_n \supset \mathcal{\bar{A}}_n$} \be
\mathcal{\bar{A}}_n \triangleq \{ (i,j) \in \mathcal{A}_n:
X'_{t_1}(i)=\bar{x}_{t_1},Y'_{t_1}(j)=\bar{y}_{t_1},...,X'_{t_k}(i)=\bar{x}_{t_k},Y'_{t_k}(j)=\bar{y}_{t_k}.
\}  \ee
On the same lines as \cite[Lemma 3]{nazari08}, we have
\begin{align}\label{formfirst}
|\bar{\mathcal{A}_n}| \geq
(\frac{\delta}{|\mathcal{X}||\mathcal{Y}|(2 \sigma-\delta)})^k
|\mathcal{A}_n|.
\end{align}

Define random variables $\bar{X}^n$, $\bar{Y}^n$ on $\mathcal{X}^n$
resp. $\mathcal{Y}^n$ by
\begin{equation}
Pr((\bar{X}^n,\bar{Y}^n)=(x^n(i),y^n(j))=\frac{1}{|\bar{\mathcal{A}}_n|}
\text{if } (i,j) \in \bar{\mathcal{A}}_n.
\end{equation}
If we denote $\bar{X}^n=(\bar{X}_1,...,\bar{X}_n)$,
$Y^n=(\bar{Y}_1,...,\bar{Y}_n)$, the Fano-distribution of the graph
$\bar{\mathcal{A}}_n$ induces a distribution
$P_{\bar{X}_t,\bar{Y}_t}$ on the random variables
$\bar{X}_t,\bar{Y}_t, t=1,\ldots,n.$ One can show that \be
\label{eq:bardash} P(\bar{X}_t=x,\bar{Y}_t=y) =
P({X}'_t=x,\bar{Y}'_t=y|X'_{t_1}(i)=\bar{x}_{t_1},Y'_{t_1}(j)=\bar{y}_{t_1},...,X'_{t_k}(i)=\bar{x}_{t_k},Y'_{t_k}(j)=\bar{y}_{t_k}),
\: \forall t. \ee

Using \eqref{eq:bardash} in Lemma \ref{lem:wringing}, we get the
bound $I(\bar{X}_t;\bar{Y}_t)< \delta$. Applying Pinsker's
inequality for I-divergences \cite{fedTop03}, we have
\begin{flalign} \label{eq:nearly_ind}
\sum_{x,y}
|&Pr(\bar{X}_t=x,\bar{Y}_t=y)-Pr(\bar{X}_t=x)Pr(\bar{Y}_t=y)|\leq 2
\delta^{1/2}, \quad 1 \leq t \leq n.
\end{flalign}

Also define
\begin{subequations}
\begin{flalign}
&\mathcal{\bar{C}}(i)=\{(i,j):(i,j) \in \mathcal{\bar{A}}_n, 1 \leq j \leq \Delta_Y^n\}. \\
&\mathcal{\bar{B}}(j)=\{(i,j):(i,j) \in \mathcal{\bar{A}}_n, 1 \leq
i \leq \Delta_X^n\}.
\end{flalign}
\end{subequations}

We are now ready to present the final lemma required to complete the
proof of the converse.
\begin{lemma}
\label{lem:sizeXY}
\begin{gather*}
R_X \leq \frac{1}{n}\sum_{t=1}^{n} H(\bar{X_t}|\bar{Y_t})+ \delta_{1n}\\
R_Y \leq \frac{1}{n}\sum_{t=1}^{n} H(\bar{Y_t}|\bar{X_t})+ \delta_{2n}\\
R_X +R_Y \leq \frac{1}{n}\sum_{t=1}^{n} H(\bar{X_t}\bar{Y_t})+ +
\delta_{3n}
\end{gather*}
 for some $\delta_{1n},\delta_{2n},\delta_{3n} \to 0$ and the distributions of the RV's are determined by the
Fano-distribution on the codewords $\{(x^n(i),y^n(j)):
(i,j) \in \bar{\mathcal{A}}_n\}$. 
\end{lemma}
\begin{proof}
We use a strong converse result for non-stationary discrete
memoryless channels, found in \cite{agustin}. Consider a DMC with
input $A_t$ and output $B_t$ $(t=1,\ldots,n)$, with average error
probability $\lambda \: (0 \leq \lambda<1)$. The result states that
the size of the message set $M$ is upper-bounded as \be
\label{eq:strong_con} \log M < \sum_{t=1}^n
I(A_t;B_t)+\frac{3}{1-\lambda}|\mathcal{A}|n^{1/2}, \ee where the
distributions of the RV's are determined by the Fano-distribution on
the codewords.

We apply the above result to three noiseless DMCs ($B_t=A_t,
\lambda=0$) as follows. Fix $\bar{Y}^n=y^n(j)$ for some $j \in
\bar{\mathcal{A}}_n$ and let the input be $\bar{X}_t, \:
t=1,\cdots,n$. Then, from \eqref{eq:strong_con} we have \be
\label{eq:bjsize} \log |\bar{\mathcal{B}}(j)| \leq \sum_{t=1}^{n}
H(\bar{X}_t|\bar{Y}_t={y}_{t}(j)) + {3}|\mathcal{X}|n^{1/2}. \ee
Similarly,
\begin{gather}
\log |\bar{\mathcal{C}}(i)| \leq \sum_{t=1}^{n} H(\bar{Y}_t|\bar{X}_t={x}_{t}(i)) + {3}|\mathcal{Y}|n^{1/2},\\
\log |\bar{\mathcal{A}}_n| \leq \sum_{t=1}^{n} H(\bar{X}_t\bar{Y}_t)
+ {3}|\mathcal{X}||\mathcal{Y}|n^{1/2}. \label{eq:jointsize}
\end{gather}
Noting that $Pr(\bar{Y_t}=y)=|\bar{\mathcal{A}}|^{-1}
\sum_{\substack{(i,j) \in \bar{\mathcal{A}}_n}} 1_{\{y_{t}(j),y\}}$,
we can sum both sides of \eqref{eq:bjsize}  over all $(i,j) \in
\bar{\mathcal{A}}_n$ to obtain \be \label{eq:sum_bj_bound}
|\bar{\mathcal{A}}_n|^{-1} \sum_{(i,j) \in \bar{\mathcal{A}}_n} \log
|\bar{\mathcal{B}}(j)| \leq \sum_{t=1}^{n} H(\bar{X}_t|\bar{Y}_t)
+{3}|\mathcal{X}|n^{1/2}. \ee
Define
\begin{equation}
B^* \triangleq \frac{2^{-2n\delta_n}}{n}
\frac{\Delta_X^n}{(n+1)^{|\mathcal{X}||\mathcal{Y}|}}(\frac{\delta}{|\mathcal{X}||\mathcal{Y}|(2
\sigma-\delta)})^k.
\end{equation}
Then,
\begin{align} \label{eq:int_sum}
|\bar{\mathcal{A}}_n|^{-1} \sum_{(i,j) \in \bar{\mathcal{A}}_n} \log
|\bar{\mathcal{B}}(j)|&= |\bar{\mathcal{A}}_n|^{-1} \sum_{j}
|\bar{\mathcal{B}}(j)| \log |\bar{\mathcal{B}}(j)| \nonumber\\&\geq
|\bar{\mathcal{A}}_n|^{-1} \sum_{\substack{j:|\bar{\mathcal{B}}(j)|
\geq B^*}}|\bar{\mathcal{B}}(j)| \log
|\bar{\mathcal{B}}(j)|\nonumber\\
&\geq |\bar{\mathcal{A}}_n|^{-1} \log(B^*)
\sum_{\substack{j:|\bar{\mathcal{B}}(j)| \geq B^*}}
|\bar{\mathcal{B}}(j)|\nonumber\\
&\geq |\bar{\mathcal{A}}_n|^{-1} \log(B^*)(|\bar{\mathcal{A}}_n|-
\Delta_Y^n B^*).
\end{align}
Combining \eqref{formfirst}, \eqref{eq:an_size} and the definition
of $B^*$, we also have
\begin{equation}
\Delta_Y^n B^* \leq \frac{1}{n} |\mathcal{\bar{A}}_n|.
\end{equation}
Using this in \eqref{eq:int_sum}, we have
\begin{flalign}
|\bar{\mathcal{A}_n}|^{-1} \sum_{(i,j) \in \bar{\mathcal{A}}_n} \log
|\bar{\mathcal{B}}(j)| & \geq |\bar{\mathcal{A}_n}|^{-1} \log(B^*)(|\bar{\mathcal{A}_n}|-\frac{1}{n}|\bar{\mathcal{A}_n}|)\nonumber\\
& = (1-\frac{1}{n}) \log(\frac{2^{-2n\delta_n}}{n}
\frac{\Delta_X^n}{(n+1)^{|\mathcal{X}||\mathcal{Y}|}}(\frac{\delta}{|\mathcal{X}||\mathcal{Y}|(2
\sigma-\delta)})^k).\label{23}
\end{flalign}
Using \eqref{eq:sum_bj_bound} in the above we have \be
\label{eq:delx_size} \log \Delta_X^n \leq
\frac{n}{n-1}(\sum_{t=1}^{n} H(\bar{X_t}|\bar{Y_t})
+{3}|\mathcal{X}|n^{1/2})+2n\delta_n+\log n +
|\mathcal{X}||\mathcal{Y}| \log(n+1)+k
\log(\frac{|\mathcal{X}||\mathcal{Y}| 2 \sigma}{\delta}) \ee
Analogously, \be\label{eq:dely_size} \log \Delta_Y^n \leq
\frac{n}{n-1}(\sum_{t=1}^{n} H(\bar{Y_t}|\bar{X_t})
+{3}|\mathcal{Y}|n^{1/2})+2n\delta_n+\log n +
|\mathcal{X}||\mathcal{Y}| \log(n+1)+k
\log(\frac{|\mathcal{X}||\mathcal{Y}| 2 \sigma}{\delta}) \ee Next,
we find an upper bound for  $\log \Delta_X^n \Delta_Y^n $. From
\eqref{formfirst}, we get
\begin{flalign}
\log |\bar{\mathcal{A}}_n| & \geq \log |\mathcal{A}_n|+ k
\log(\frac{\delta}{|\mathcal{X}||\mathcal{Y}|(2 \sigma-\delta)})\nonumber\\
&\geq \log |\mathcal{A}_n|+ k
\log(\frac{\delta}{|\mathcal{X}||\mathcal{Y}|2 \sigma})\nonumber\\
&= \log |\mathcal{A}_n|- k \log(\frac{2 \sigma}{\delta})-k\log(|\mathcal{X}||\mathcal{Y}|) \nonumber \\
&\stackrel{(a)}{\geq} \log(\Delta_X^n \Delta_Y^n) -
|\mathcal{X}||\mathcal{Y}| \log(n+1) -2n\delta_n -k
\log(\frac{|\mathcal{X}||\mathcal{Y}|2\sigma}{\delta}),
\end{flalign}
where $(a)$ is obtained by using \eqref{eq:an_size}. Using
\eqref{eq:jointsize}, the above inequality becomes \be
\label{eq:delxdely_size} \log(\Delta_X^n \Delta_Y^n) \leq
\sum_{t=1}^{n} H(\bar{X}_t\bar{Y}_t) +
{3}|\mathcal{X}||\mathcal{Y}|n^{1/2}+ |\mathcal{X}||\mathcal{Y}|
\log(n+1)+ 2n\delta_n + k \log(\frac{2\sigma}{\delta})+k
\log(|\mathcal{X}||\mathcal{Y}|) \ee Using the lower bounds on the
sizes of $\Delta_X, \Delta_Y$ from \ref{def:nc_subg}, we can rewrite
\eqref{eq:delx_size},\eqref{eq:dely_size} and
\eqref{eq:delxdely_size} as
\begin{gather}
R_X-\delta_n \leq \frac{1}{n-1}\sum_{t=1}^{n} H(\bar{X_t}|\bar{Y_t})
+{3|\mathcal{X}|}\frac{n^{1/2}}{n-1}+2\delta_n + \frac{\log n + |\mathcal{X}||\mathcal{Y}| \log(n+1)}{n} +\frac{k}{n}\log(\frac{2|\mathcal{X}||\mathcal{Y}|\sigma}{\delta}) \label{eq:rx_size}\\
R_Y -\delta_n \leq \frac{1}{n-1}\sum_{t=1}^{n}
H(\bar{Y_t}|\bar{X_t})
+{3|\mathcal{Y}|}\frac{n^{1/2}}{n-1}+2\delta_n+ \frac{\log n +
|\mathcal{X}||\mathcal{Y}| \log(n+1)}{n}
+\frac{k}{n} \log(\frac{2|\mathcal{X}||\mathcal{Y}| \sigma}{\delta})\\
R_X + R_Y -2\delta_n \leq \frac{1}{n}\sum_{t=1}^{n}
H(\bar{X}_t\bar{Y}_t) +
{3}|\mathcal{X}||\mathcal{Y}|\frac{n^{1/2}}{n-1}+
|\mathcal{X}||\mathcal{Y}|\frac{\log(n+1)}{n}+ 2\delta_n +
\frac{k}{n} \log(\frac{2|\mathcal{X}||\mathcal{Y}|\sigma}{\delta})
\label{rxry_size}
\end{gather}
For our proof we would like all the terms on the right hand side of
the above equations (except the entropies) to converge to 0 as $n\to
\infty$. This will happen if
\[ \frac{k}{n}\log(\frac{2\sigma}{\delta}) \to 0. \]
Recall from Lemma \ref{lem:ixybound} that $\sigma= 2n\delta_n +
|\mathcal{X}||\mathcal{Y}|\log(n+1)$ and $k<\frac{2\sigma}{\delta}$.
Hence we need to choose $\delta$ such that \be  \label{eq:asymp}
\frac{2\sigma}{n\delta}\log(\frac{2\sigma}{\delta})\sim
 \frac{\delta_n + \frac{\log n}{n}}{\delta} \left(\log(n\delta_n+\log n)-\log\delta\right) \to 0.
\ee
>From our assumption in the beginning, we have  $ \delta_n \log n \to 0$. Set
\be \label{eq:delta_value} \delta= (\delta_n \log n)^{1/2} \ee We
see that asymptotically, \eqref{eq:asymp} becomes \be
\label{eq:asymp1} \frac{\delta_n^{1/2}}{(\log n)^{1/2}}\left[\log
(n\delta_n+\log n) -\log(\delta_n^{1/2})-\log \log n\right] \ee We
separately consider each of the terms in the equation above
\begin{enumerate}
\item If $\log (n\delta_n+\log n) \sim \log (n\delta_n)$ for large $n$, then
\be
\begin{split}
\frac{\delta_n^{1/2}}{(\log n)^{1/2}}\log(n\delta_n+\log n) &\sim \frac{\delta_n^{1/2}}{(\log n)^{1/2}}\log (n\delta_n) =\frac{\delta_n^{1/2}}{(\log n)^{1/2}}\left[\log n+ \log \delta_n\right]\\
&=(\delta_n\log n)^{1/2}+\frac{\delta_n^{1/2}\log \delta_n}{(\log
n)^{1/2}} \to 0, \text{ since } \delta_n \to 0.
\end{split}
\ee
%
If $\log (n\delta_n+\log n) \sim \log (\log n)$ for large $n$, then
\be \frac{\delta_n^{1/2}}{(\log n)^{1/2}}\log(n\delta_n+\log n) \sim
\frac{\delta_n^{1/2}}{(\log n)^{1/2}}\log (\log n) \to 0. \ee
\item $\frac{\delta_n^{1/2}}{(\log n)^{1/2}} \log(\delta_n^{1/2})  \to 0 \text { because } x \log x \to 0 \text{ when } x \to 0.$
\item $\frac{\delta_n^{1/2}}{(\log n)^{1/2}} \log \log n = (\delta_n \log n)^{1/2}\frac{\log \log n}{\log n} \to 0.$
\end{enumerate}
Hence the term in \eqref{eq:asymp1} converges to $0$ as $n \to
\infty$, completing the proof of the lemma.
\end{proof}

We can rewrite Lemma \ref{lem:sizeXY} using new variables
$\bar{X},\bar{Y},Q$, where $Q=t \in \{1,2,...,n\}$ with probability
$\frac{1}{n}$ and $P_{\bar{X},\bar{Y}|Q=t}=P_{\bar{X}_t,\bar{Y}_t}$.
So we now have (for all sufficiently large $n$),
\begin{gather}
R_X \leq H(\bar{X}|\bar{Y},Q)+ \delta_{1n} \label{eq:rx_q}\\
R_Y \leq  H(\bar{Y}|\bar{X}, Q)+ \delta_{2n}\\
R_X + R_Y \leq  H(\bar{X},\bar{Y}|Q) + \delta_{3n}, \label{eq:rxy_q}
\end{gather}
for some $\delta_{1n},\delta_{2n}, \delta_{3n} \to 0$.

Finally, using \eqref{eq:nearly_ind}, we also have
\begin{equation}
\begin{split}
&|Pr(\bar{X}=x,\bar{Y}=y|Q=t)-Pr(\bar{X}=x|Q=t)Pr(\bar{Y}=y|Q=t)|\\
&=|Pr(\bar{X}_t=x,\bar{Y}_t=y)-Pr(\bar{X}_t=x)Pr(\bar{Y}_t=y)| \\
& \leq 2\delta^{1/2} = 2(\delta_n \log n)^{1/4} \to 0 \text{ as } n
\to \infty.
\end{split}
\end{equation}
In other words, for all $t$, $\bar{X}_t,\bar{Y}_t$ are almost
independent for large $n$. Consequently, using the continuity of
mutual information with respect to the joint distribution, Lemma
\ref{lem:sizeXY} holds with for any joint distribution
$P_{Q}P_{\bar{X}|Q}P_{\bar{Y}|Q}$ such that the marginal on
$(\bar{X},\bar{Y})$ is $P_{\bar{X},\bar{Y}}$. Recall that
$P_{\bar{X},\bar{Y}}$ is the dominant joint type that is
$\lambda_n$-close to $P_{X,Y}$. Using suitable continuity arguments,
we can now argue that Lemma \ref{lem:sizeXY} holds with for any
joint distribution $P_{Q}P_{{X}|Q}P_{{Y}|Q}$ such that the marginal
on $({X},{Y})$ is $P_{{X},{Y}}$, completing the proof of the
converse.

\subsection{Nearly Empty Subgraphs} \label{subsec:subgraphsnearlyempty}

So far, we have discussed properties of subgraphs of the typicality
graph $G_n(\epsilon_{1n}, \epsilon_{2n}, \lambda_n)$ such as the
containment of nearly complete subgraphs and subgraphs of general
degree. Now, we turn our attention to the presence of nearly empty
subgraphs in the typicality graph. Our approach towards this problem
differs slightly from the approach we took in Sections
\ref{subsec:subgraphsgendegree} and
\ref{subsec:subgraphsnearlycomplete}. While in these sections we
characterized the subgraphs based on the degrees of their vertices,
in this section we would characterize nearly empty subgraphs by the
total number of edges present in such graphs. To effect this
characterization, we take a different approach than the one used in
previous sections and analyze the probability that a randomly chosen
subgraph of the typicality graph has far fewer edges than expected.
In particular, we focus attention on the case when the random
subgraph has no edges.

Consider a pair $(X,Y)$ of discrete memoryless stationary correlated
sources with finite alphabets $\mathcal{X}$ and $\mathcal{Y}$
respectively. Suppose we sample $2^{nR_1}$ sequences from the
typical set $T_{\epsilon_{1n}}^n(X)$ of $X$ independently with
replacement and similarly sample $2^{nR_2}$ sequences from the
typical set $T_{\epsilon_{2n}}^n(Y)$ of $Y$. The underlying
typicality graph $G_n(\epsilon_{1n},\epsilon_{2n},\lambda_{n})$
induces a bipartite graph on these $2^{nR_1} + 2^{nR_2}$ sequences.
We provide a characterization of the probability that this graph is
sparser than expected. This characterization is obtained through the
use of a version of Suen's inequalities \cite{janson} and the Lovasz
local lemma \cite{alon-spencer} listed below.

\begin{lemma} \cite{janson} \label{lemma:suen}
Let $I_i \in \mbox{Be}(p_i), i \in \mathcal{I}$ be a family of
Bernoulli random variables. Their dependency graph $L$ is formed in
the following manner. Denote the random variable $I_i$ by a vertex
$i$ and join vertices $i$ and $j$ by an edge if the corresponding
random variables are dependent. Let $X = \sum_i \mathbb{E}(I_i)$ and
$\Gamma = \mathbb{E}(X) = \sum_i p_i$. Moreover, write $i \sim j$ if
$(i,j)$ is an edge in the dependency graph $L$ and let $\Theta =
\frac{1}{2} \sum_i \sum_{j \sim i} \mathbb{E}(I_i I_j)$ and $\theta
= \max_i \sum_{j \sim i} p_j$. Then, Suen's inequalities state that
for any $0 \leq a \leq 1$,
\begin{equation}
P(X \leq a \Gamma) \leq \exp \left\{ - \min \left( (1-a)^2
\frac{\Gamma^2}{8\Theta + 2\Gamma}, (1-a) \frac{\Gamma}{6 \theta}
\right) \right\} \label{eq:sueneq2}
\end{equation}
Putting $a = 0$, this can be further tightened to
\begin{equation}
P(X = 0) \leq \exp \left\{ - \min \left( \frac{\Gamma^2}{8 \Theta},
\frac{\Gamma}{2}, \frac{\Gamma}{6 \theta} \right) \right\}
\label{eq:sueneq1}
\end{equation}
\end{lemma}

\begin{lemma} \cite{alon-spencer} \label{lemma:lovasz}
Let $L$ be the dependency graph for events $\varepsilon_1, \dots,
\varepsilon_n$ in a probability space and let $E(L)$ be the edge set
of $L$. Suppose there exists $x_i \in [0,1], 1 \leq i \leq n$ such
that
\begin{equation}
P(\varepsilon_i) \leq x_i \prod_{(i,j) \in E(L)} (1-x_j).
\end{equation}
Then, we have
\begin{equation} \label{eq:lovaszeq1}
P(\cap_{i=1}^{n} \overline{\varepsilon_i}) \geq \prod_{i=1}^n
(1-x_i).
\end{equation}
Another version of the local lemma is as given below. Let $\phi(x),
0 \leq x \leq e^{-1}$ be the smallest root of the equation $\phi(x)
= e^{x\phi(x)}$. With definitions of $\Gamma$ and $\theta$ as in
Lemma \ref{lemma:suen} and defining $\tau \triangleq \max_i
P(\varepsilon_i)$, we have
\begin{equation} \label{eq:lovaszeq2}
P\left(\cap_{i=1}^{n} \overline{\varepsilon_i} \right) \geq \exp
\left\{- \Gamma \phi(\theta + \tau) \right\}
\end{equation}
\end{lemma}

With these preliminaries, we are ready to state the main result of
this section.

\begin{prop} \label{prop:emptygraph}
Suppose $X$ and $Y$ are correlated finite alphabet memoryless random
variables with joint distribution $p(x,y)$. Let
$\epsilon_{1n},\epsilon_{2n}, \lambda_n$ satisfy the `delta
convention' and $R_1, R_2$ be any positive real numbers such that
$R_1 + R_2 > I(X;Y)$. Let $\mathcal{C}_X$ be a collection of
$2^{nR_1}$ sequences picked independently and with replacement from
$T_{\epsilon_{1n}}^n(X)$ and let $\mathcal{C}_Y$ be defined
similarly. Let $U$ be the cardinality of the set
\begin{equation}
\mathcal{U} \triangleq \{ (x^n,y^n) \in \mathcal{C}_X \times
\mathcal{C}_Y \colon (x^n,y^n) \in T_{\lambda_n}^n(X,Y) \}
\end{equation}
Assume, without loss of generality that $R_1 \geq R_2$. Then, for
any $\gamma \geq 0$, we have
\begin{equation} \label{eq:emptygrapheq1}
\lim_{n \rightarrow \infty} \frac{1}{n} \log \log \left[ \mathbb{P}
\left( \frac{\mathbb{E}(U) - U}{\mathbb{E}(U)} \geq e^{-n\gamma}
\right) \right]^{-1} \geq \left\{ \begin{array}{cc} R_1 + R_2 -
I(X;Y) - \gamma & \mbox{if } R_1 < I(X;Y) \\ R_2 - \gamma & \mbox{if
} R_1 \geq I(X;Y) \end{array} \right.
\end{equation}
Setting $\gamma = 0$ in the above equation gives us
\begin{equation} \label{eq:emptygrapheq2}
\lim_{n \rightarrow \infty} \frac{1}{n} \log \log
\frac{1}{\mathbb{P}(U = 0)} \geq \min \left( R_2, R_1 + R_2 - I(X;Y)
\right)
\end{equation}
This inequality holds with equality when $R_2 \leq R_1 \leq I(X;Y)$.
\end{prop}

\begin{proof}
Let $X^n(i)$ and $Y^n(j)$ denote the $i$th and $j$th codewords in
the random codebooks $\mathcal{C}_X$ and $\mathcal{C}_Y$
respectively. For $1 \leq i \leq 2^{nR_1}$ and $1 \leq j \leq
2^{nR_2}$, define the indicator random variables
\begin{equation}
U_{ij} \triangleq \left\{ \begin{array}{cc} 1 & \mbox{if }
(X^n(i),Y^n(j)) \in T_{\lambda_n}^n(X,Y) \\ 0 & \mbox{else}
\end{array} \right.
\end{equation}
The cardinality of the set $\mathcal{U}$ is then
\begin{equation}
U = \sum_{i=1}^{2^{nR_1}} \sum_{j=1}^{2^{nR_2}} U_{ij}
\end{equation}
We derive upper bounds on the probability of the lower tail of $U$
using Suen's inequality. To do this, we first set up the dependency
graph of the indicator random variables $U_{ij}$. The vertex set of
the graph is indexed by the ordered pair $(i,j), 1 \leq i \leq
2^{nR_1}, 1 \leq j \leq 2^{nR_2}$. From the nature of the random
experiment, it is clear that the indicator random variables $U_{ij}$
and $U_{i'j'}$ are independent if and only if $i \neq i'$ and $j
\neq j'$. Thus, each vertex $(i,j)$ is connected to exactly
$2^{nR_1} + 2^{nR_2} - 2$ vertices of the form $(i,j'), j' \neq j$
or $(i',j), i' \neq i$. If vertices $(i,j)$ and $(k,l)$ are
connected, we denote it by $(i,j) \sim (k,l)$.

In order to estimate $\Gamma, \Theta$ and $\theta$ as defined in
Lemma \ref{lemma:suen}, define the following quantities. Let
$\alpha_{ij} \triangleq \mathbb{P}(U_{ij} = 1)$ and
$\beta_{\{ij\}\{kl\}} \triangleq \mathbb{E}(U_{ij}U_{kl})$ where
$(i,j) \sim (k,l)$. Using Facts 1 and 2, uniform bounds can be
derived for these quantities as
\begin{equation} \label{eq:alphadefeq}
\alpha \triangleq 2^{-n(I(X;Y) + \epsilon_{3n})} \leq \alpha_{ij}
\leq 2^{-n(I(X;Y) - \epsilon_{3n})} \triangleq \alpha^{'}
\end{equation}
where $\epsilon_{3n}$ is a continuous positive function of
$\epsilon_{1n}, \epsilon_{2n}$ and $\lambda_{n}$ that goes to $0$ as
$n \rightarrow \infty$. Similarly, a uniform bound on
$\beta_{\{ij\}\{kl\}}$ can be derived as
\begin{equation}
2^{-2n(I(X;Y) + 2 \epsilon_{4n})} \leq \beta_{\{ij\}\{kl\}} \leq
2^{-2n(I(X;Y) - 2\epsilon_{4n})} \triangleq \beta
\end{equation}
where $\epsilon_{4n}$ is a continuous positive function of
$\epsilon_{1n}, \epsilon_{2n}$ and $\lambda_{n}$ that goes to $0$ as
$n \rightarrow \infty$.

The quantities involved in Suen's inequality can now be estimated.
\begin{align}
\Gamma &\triangleq \mathbb{E}(U) \geq 2^{n(R_1+R_2)} \alpha \\
\Theta &\triangleq \frac{1}{2} \sum_{(i,j)} \sum_{(k,l) \sim (i,j)} \mathbb{E}(U_{ij}U_{kl}) \leq \frac{1}{2} 2^{n(R_1 + R_2)}(2^{nR_1} + 2^{nR_2} - 2) \beta \\
\theta &\triangleq \max_{(i,j)} \sum_{(k,l) \sim (i,j)}
\mathbb{E}(U_{kl}) \leq (2^{nR_1} + 2^{nR_2} - 2) \alpha^{'}
\end{align}
Substituting these bounds into equations (\ref{eq:sueneq1}) and
(\ref{eq:sueneq2}) proves the claims made in equations
(\ref{eq:emptygrapheq1}) and (\ref{eq:emptygrapheq2}) of Proposition
\ref{prop:emptygraph}.

A lower bound to the probability of the induced random subgraph
being empty can be derived by employing the Lovasz local lemma on
the $2^{n(R_1 + R_2)}$ events $\{U_{ij} = 1\}, 1 \leq i \leq
2^{nR_1}, 1 \leq j \leq 2^{nR_2}$. Symmetry considerations imply
that all $x_i$ can be set identically to $x$ in Lemma
\ref{lemma:lovasz}. Then the local lemma states that if there exists
$x \in [0,1]$ such that $\alpha \leq P(U_{ij} = 1) \leq x
(1-x)^{(2^{nR_1} + 2^{nR_2}-2})$, then $P(U=0) \geq (1-x)^{2^{n(R_1
+ R_2)}}$. It is easy to verify that for such an $x$ to exist, we
need $R_2 \leq R_1 < I(X;Y)$ and if so, $x = 2^{-nR_1}$ satisfies
the condition. Therefore, we have
\begin{equation} \label{eq:lowerboundeq1}
P(U = 0) \geq \exp \left( -\left(2^{nR_2} + 1 \right) \right) \qquad
R_2 \leq R_1 < I(X;Y)
\end{equation}
We can derive a similar bound using the second version of the local
lemma given in Lemma \ref{lemma:lovasz}. While $\Gamma$ and $\theta$
are same as estimated earlier, $\tau = \max_{(i,j)} P(U_{ij} = 1)$
is upper bounded by $\alpha^{'}$ as defined in equation
(\ref{eq:alphadefeq}). Hence,
\begin{equation} \label{eq:lowerboundeq2}
P(U = 0) \geq \exp \left(-\Gamma \phi(\theta + \tau) \right).
\end{equation}
Under the same assumption $R_2 \leq R_1 < I(X;Y)$, $\theta + \tau
\leq (2^{nR_1} + 2^{nR_2} - 2) \alpha^{'} \rightarrow 0$ as $n
\rightarrow \infty$ and hence $\phi(\theta + \tau) \rightarrow 1$.
Combining equations (\ref{eq:lowerboundeq1}) and
(\ref{eq:lowerboundeq2}), taking logarithms and letting $n
\rightarrow \infty$, we get
\begin{equation} \label{eq:lowerboundeqfinal}
\lim_{n \rightarrow \infty} \frac{1}{n} \log \log \frac{1}{P(U=0)}
\leq \min \left(R_2, R_1 + R_2 - I(X;Y) \right).
\end{equation}
Comparing this to equation (\ref{eq:emptygrapheq2}) shows that this
expression is asymptotically tight in the regime $R_2 \leq R_1 <
I(X;Y)$.
\end{proof}

\bibliographystyle{plain}
\bibliography{typ_graphs}

\end{document}